\documentclass[11pt]{article}

\usepackage{amsmath}
\usepackage{amssymb}
\usepackage{amsfonts}
\usepackage{amsthm} 
\usepackage{mathrsfs}
\usepackage{color}
\usepackage{url}
\usepackage{comment}
\usepackage{xspace}

\usepackage{graphicx}

\usepackage{cite}
\usepackage{pifont}
\newcommand{\cmark}{\ding{51}}%
\newcommand{\xmark}{\ding{55}}%

\usepackage{algorithm}
\usepackage[noend]{algpseudocode}

\newtheorem{theorem}{Theorem}
\newtheorem{lemma}[theorem]{Lemma}
\newtheorem{proposition}[theorem]{Proposition}

\newtheorem{claim}[theorem]{Claim}

\newtheorem{example}[theorem]{Example}

\newcommand{\PROPavg}{\textsf{PROPavg}\xspace}
\newcommand{\PROPm}{\textsf{PROPm}\xspace}
\newcommand{\PROPx}{\textsf{PROPx}\xspace}
\newcommand{\PROP}{\textsf{PROP}\xspace}
\newcommand{\EFX}{\textsf{EFX}\xspace}
\newcommand{\EF}{\textsf{EF}\xspace}

\title{ Proportional Allocation of Indivisible Goods \\ up to the Least Valued Good on Average\thanks{An earlier version of this work was presented in ISAAC 2022\cite{kobayashi}.} }

\author{
\begin{tabular}[h]{cc}  
Yusuke Kobayashi\thanks{Research Institute for Mathematical Sciences, Kyoto University, Japan.
E-mail: yusuke@kurims.kyoto-u.ac.jp} \and Ryoga Mahara\thanks{Department of Mathematical Informatics, University of Tokyo.
E-mail: mahara@mist.i.u-tokyo.ac.jp}
\end{tabular}
}

\date{}
\begin{document}
\maketitle
\begin{abstract}
We study the problem of fairly allocating a set of indivisible goods to multiple agents and focus on the proportionality, which is one of the classical fairness notions.
Since proportional allocations do not always exist when goods are indivisible, approximate concepts of proportionality have been considered in the previous work.
Among them, proportionality up to the maximin good (\PROPm) has been the best approximate notion of proportionality 
that can be achieved for all instances~\cite{ijcai2021-4}. 
In this paper, we introduce the notion of {\it proportionality up to the least valued good on average} (\PROPavg), which is a stronger notion than \PROPm, and show that a \PROPavg allocation always exists for all instances and can be computed in polynomial time. 
Our results establish \PROPavg as a notable non-trivial fairness notion that can be achieved for all instances.
Our proof is constructive, and based on a new technique that 
generalizes the cut-and-choose protocol and uses a recursive technique. 
\end{abstract}

\thispagestyle{empty}
\newpage
\setcounter{page}{1}

\section{Introduction} 
\label{sec: intro}

\subsection{Proportional Allocation of Indivisible Goods}

We study the problem of fairly allocating a set of indivisible goods to multiple agents under additive valuations.
Fair division of indivisible goods is a fundamental and well-studied problem in Economics and Computer Science.
We are given a set $M$ of $m$ indivisible goods and a set $N$ of $n$ agents with individual valuations.
Under additive valuations, each agent $i\in N$ has a value $v_i(\{g\})\ge 0$ for each good $g$ and her value for a bundle $S$ of goods is equal to the sum of the values of each good $g\in S$, i.e., $v_i(S)=\sum_{g\in S} v_i(\{g\})$.
An indivisible good can not be split among multiple agents and this causes finding a fair division to be a difficult task.

One of the standard notions of fairness is {\it proportionality (\PROP)}.
Let an {\it allocation} $X=(X_1,X_2,\ldots, X_n)$ denote a partition of $M$ to $N$ into $n$ bundles such that $X_i$ is allocated to agent $i$.
An allocation $X$ is {\it proportional} if $v_i(X_i)\ge \frac{1}{n}v_i(M)$ holds for each agent $i$.
In other words, in a proportional allocation, every agent receives a set of goods whose value is at least $1/n$ fraction of the value of the entire set. 
Unfortunately, proportional allocations do not always exist when goods are indivisible.
For instance, when allocating a single indivisible good to more than one agents it is impossible to achieve any proportional allocation.
Thus, several relaxations of proportionality such as {\PROP}1, \PROPx, and \PROPm have been considered in the previous work.

Each of these notions requires that each agent $i\in N$ receives value at least $\frac{1}{n}v_i(M)-d_i(X)$, where $d_i(X)$ is a nonnegative real number appropriately defined for each notion.
{\it Proportionality up to the largest valued good} ({\PROP}1) is a relaxation of proportionality that was introduced by~Conitzer et al.~\cite{conitzer2017fair}.
{\PROP}1 requires $d_i(X)$ to be the largest value that agent $i$ has for any good allocated to other agents, i.e., $d_i(X)=\max_{k\in N\setminus \{i\}} \max_{g\in X_k} v_i(\{g\})$.
It is shown in~\cite{conitzer2017fair} that there always exists a Pareto optimal\footnote{An allocation $X=(X_1,\ldots, X_n)$ is Pareto optimal if there is no allocation $Y=(Y_1,\ldots, Y_n)$ such that $v_i(Y_i) \ge v_i(X_i)$ for any agent $i$, and there exists an agent $j$ such that $v_j(Y_j) > v_j(X_j)$.} allocation that satisfies {\PROP}1.
Moreover, Aziz et al.~\cite{aziz2020polynomial} presented a polynomial-time algorithm that finds a {\PROP}1 and Pareto optimal allocation even in the presence of chores, i.e., some items can have negative value.

Another relaxation is {\it proportionality up to the least valued good} (\PROPx), which is much stronger than ${\PROP}1$. 
\PROPx requires $d_i(X)$ to be the least value that agent $i$ has for any good allocated to other agents, i.e., $d_i(X)=\min_{k\in N\setminus \{i\}} \min_{g\in X_k} v_i(\{g\})$. 
Moulin~\cite{moulin2019fair} gave an example for which no \PROPx allocation exists, and Aziz et al.~\cite{aziz2020polynomial} gave a simpler example. 

Recently, Baklanov et al.~\cite{baklanov2021achieving} introduced {\it proportionality up to the maximin good} (\PROPm). 
\PROPm requires $d_i(X)=\max_{k\in N\setminus \{i\}} \min_{g\in X_k} v_i(\{g\})$, 
which shows that \PROPm is the notion between {\PROP}1 and \PROPx. 
It is shown in~\cite{baklanov2021achieving} that a \PROPm allocation always exists for instances with at most five agents, and later 
Baklanov et al.~\cite{ijcai2021-4} showed that there always exists a \PROPm allocation for any instance and it can be computed in polynomial time. 
To the best of our knowledge, \PROPm has been the best approximate notion of proportionality that is shown to exist for all instances. 

However, \PROPm is not a good enough relaxation of proportionality in some cases. 
For example, suppose that there exists a good $g\in M$ for which every agent has value at least $1/n$ fraction of the value of $M$. 
Then allocating $g$ to some agent $i$ and allocating all the goods in $M\setminus \{g\}$ to another agent achieves a \PROPm allocation, 
whereas it will be better to allocate $M\setminus \{g\}$ to $N \setminus \{i\}$ in a fair manner (see Example~\ref{example1}). 
This motivates the study of better relaxations of proportionality than \PROPm.

\subsection{Our Contribution}
\label{sec: results}

In this paper, we introduce {\it proportionality up to the least valued good on average (\PROPavg)}, a new relaxation of proportionality, and show that there always exists a \PROPavg allocation for all instances and can be computed in polynomial time.
\PROPavg requires $d_i(X)$ to be the average of minimum values that agent $i$ has for any good allocated to other agents, i.e., $d_i(X)=\frac{1}{n-1}\sum_{k\in N\setminus \{i\}} \min_{g\in X_k} v_i(\{g\})$.
It is easy to see that \PROPavg implies \PROPm.
Note that a similar and slightly stronger notion was introduced by Baklanov et al.~\cite{baklanov2021achieving} with the name of Average-\EFX (\textsf{Avg-EFX}), 
where $d_i(X)=\frac{1}{n}\sum_{k\in N\setminus \{i\}} \min_{g\in X_k} v_i(\{g\})$. 
It remains open whether an \textsf{Avg-EFX} allocation always exists even in the case of four agents. 
The following example demonstrates that \PROPavg (or \textsf{Avg-EFX}) is a reasonable relaxation of proportionality compared to \PROPm.

\begin{example}\label{example1}
Suppose that $N=\{1, 2, 3\}$, $M=\{g_1, g_2, g_3, g_4\}$, and each agent has an identical additive valuation defined as follows: $v(\{g_1\})=10, v(\{g_2\})=v(\{g_3\})=7, v(\{g_4\})=6$.

Table~\ref{table:example} summarizes whether several allocations satisfy or not each fairness notions in this instance.
As $v(\{g_1\}) \ge 10$, the allocation $(\{g_1, g_2, g_3, g_4\}, \emptyset, \emptyset)$ satisfies {\PROP}1 even though agent $2$ and $3$ receive no good.
Similarly, the allocation $(\{g_1\}, \{g_2, g_3, g_4\}, \emptyset)$ satisfies \PROPm even though agent $3$ receives no good.
It is easy to see that every agent has to receive at least one good in any \PROPavg allocation in this instance.


\end{example}

\begin{table}[bth]
 \caption{The relationship between several allocations and fairness notions in Example~\ref{example1}. The symbol ``\cmark'' indicates that the corresponding allocation satisfies the corresponding fairness. The symbol ``\xmark'' indicates that the corresponding allocation does not satisfy the corresponding fairness. See Section~\ref{sec: related} for the definition of EFX.}
 \label{table:example}
 \centering
  \begin{tabular}{c|c|c|c|c}\hline
	    & \EFX & \PROPavg & \PROPm & {\PROP}1 \\ \hline
	$(\{g_1\}, \{g_2, g_4\}, \{g_3\})$ &  \cmark & \cmark & \cmark & \cmark \\ \hline
	$(\{g_1\}, \{g_2, g_3\}, \{g_4\})$ &   \xmark & \cmark & \cmark & \cmark \\ \hline
	$(\{g_1\}, \{g_2, g_3, g_4\}, \emptyset)$ &  \xmark & \xmark & \cmark & \cmark \\ \hline
	$(\{g_1, g_2, g_3, g_4\}, \emptyset, \emptyset)$ & \xmark & \xmark & \xmark & \cmark \\ \hline
  \end{tabular}
\end{table}

The main contribution of this paper is the following theorem which extends the results on \PROPm allocations shown by Baklanov et al.~\cite{ijcai2021-4}. 

\begin{theorem}\label{thm: main}
There always exists a \PROPavg allocation when each agent has a non-negative additive valuation.
Furthermore, it can be computed in polynomial time.
\end{theorem}

Known results on relaxations of proportionality are summarized in Table~\ref{table:01}. 

\begin{table}
 \caption{Relaxations of Proportionality}
 \label{table:01}
 \centering
  \begin{tabular}{llll}
   \hline
	    & $d_i(X)$ & Does it always exist? & Time complexity\\
   \hline 
   \PROPx & $\min_{k\in N\setminus \{i\}} \min_{g\in X_k} v_i(\{g\})$ &  No \cite{moulin2019fair,aziz2020polynomial} & --\\
   \textsf{Avg-EFX} & $\frac{1}{n}\sum_{k\in N\setminus \{i\}} \min_{g\in X_k} v_i(\{g\})$ & Open & Open\\
   \PROPavg &  $\frac{1}{n-1}\sum_{k\in N\setminus \{i\}} \min_{g\in X_k} v_i(\{g\})$ & Yes  \textbf{(our result)} & Polynomial time  \textbf{(our result)}\\
   \PROPm & $\max_{k\in N\setminus \{i\}} \min_{g\in X_k} v_i(\{g\})$ & Yes  \cite{ijcai2021-4} &Polynomial time\\
   {\PROP}1 & $\max_{k\in N\setminus \{i\}} \max_{g\in X_k} v_i(\{g\})$ & Yes \cite{aziz2020polynomial} &Polynomial time\\
   \hline
  \end{tabular}
\end{table}


\subsection{Our Techniques}
\label{sec: tech}

Our algorithm can be seen as a generalization of {\it cut-and-choose} protocol, which is a well-known procedure to fairly allocate resources between two agents.
In the cut-and-choose protocol, one agent partitions resources into two bundles for her valuation, and then the other agent chooses the best bundle of the two for her valuation.
We generalize this protocol from two agents to $n$ agents in the following way: 
some $n-1$ agents partition the goods into $n$ bundles, and then the remaining agent chooses the best bundle among them for her valuation.
To apply this protocol, it suffices to show that there exists a partition of the goods into $n$ bundles such that 
no matter which bundle the remaining agent chooses, the remaining $n-1$ bundles can be allocated to the first $n-1$ agents fairly.

In our algorithm, we find such a partition by using an auxiliary graph called {\it \PROPavg-graph}. 
A formal definition of the \PROPavg-graph is given in Section~\ref{sec:propavggraph}.
In  Section~\ref{sec:mainproof}, we show the existence of a \PROPavg allocation by constructing a pseudo-polynomial time algorithm for finding it.
In Section~\ref{sec:poly}, by improving our algorithm in Section~\ref{sec:mainproof},  we show how to find a \PROPavg allocation in polynomial time.

Let us emphasize that introducing the \PROPavg-graph is a key technical ingredient in this paper. 
It is also worth noting that 
Hall's marriage theorem~\cite{hall1935representatives}, a classical and famous result in discrete mathematics, plays an important role in our argument.

\subsection{Related Work}
\label{sec: related}
Fair division of divisible resources is a classical topic starting from the 1940's~\cite{Steinhaus} and has a long history in multiple fields such as Economics, Social Choice Theory, and Computer Science\cite{robertson1998fair, brams1996fair, brandt2016handbook, moulin2004fair}.
On the other hand, fair division of indivisible items has actively studied in recent years (see, e.g., \cite{amanatidis2022fair} for a recent survey).

In the context of fair division, besides proportionality, {\it envy-freeness (\EF)} is another well-studied notion of fairness. 
An allocation is called {\it envy-free} if for each agent, she receives a set of goods for which she has value at least value of the set of goods any other agent receives.
As in the proportionality case, envy-free allocations do not always exist when goods are indivisible, and several relaxations of envy-freeness have been considered.
Among them, a notable one is {\it envy-freeness up to one good} (\textsf{EF}1)~\cite{budish2011combinatorial}. 
It is known that there always exists an \textsf{EF}1 allocation, and it can be computed in polynomial time~\cite{lipton2004approximately}. 
Another notable relaxation is {\it envy-freeness up to any good} (\EFX)~\cite{caragiannis2019unreasonable}.
An allocation $X=(X_1,\ldots ,X_n)$ is called \EFX if for any pair of agents $i,j \in N$, $v_i(X_i) \ge v_i(X_j)-m_i(X_j)$, where $m_i(X_j)$ is the value of the least valuable good for agent $i$ in $X_j$.
It is one of the major open problems in fair division whether \EFX allocations always exist or not. 
As mentioned in~\cite{baklanov2021achieving}, it is easy to see that \EFX implies \textsf{Avg-EFX}.
As with \EFX, it is not known whether \textsf{Avg-EFX} allocations always exist for instances with four or more agents.
The relationship among notions mentioned above and the existence results are summarized in Figure~\ref{fig: 1}.
\begin{figure}[tbp]
  \begin{center}
   \includegraphics[scale=0.5]{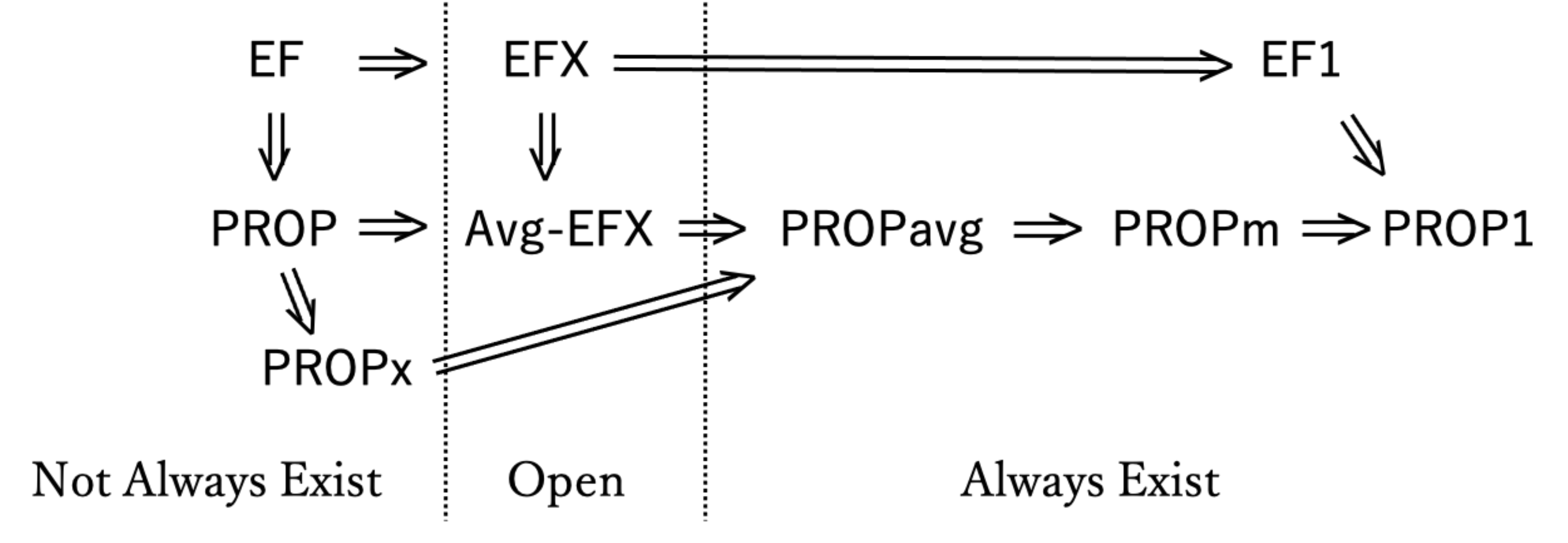}
 \end{center}
  \caption{The relationship among some fairness notions and the existence results for these fairness notions. \EF, \PROP, and \PROPx do not always exist, whereas \PROPavg, \PROPm, {\EF}1, and {\PROP}1 always exist when each agent has non-negative additive valuation. It is not known whether \EFX or \textsf{Avg-EFX} always exist or not for all instances.} 
  \label{fig: 1}
\end{figure}
There have been several studies on the existence of  an \EFX allocation for restricted cases.
Plaut and Roughgarden~\cite{plaut2020almost} showed that an \EFX allocation always exists for instances with two agents even when each agent can have more general valuations than additive valuations.
Chaudhury et al.~\cite{chaudhury2020efx} showed that an \EFX allocation always exists for instances with three agents.
It is not known whether \EFX allocations always exist even in the case of four agents having additive valuations.
There are several studies of the cases with restricted valuations.  
For example, there always exists an \EFX allocation when valuations are identical \cite{plaut2020almost}, two types \cite{mahara2020existence, mahara2021extension}, binary \cite{barman2018greedy, darmann2015maximizing}, or bi-valued \cite{amanatidis2021maximum}.

Another approach related to \EFX is a setting to allow unallocated goods, which is known as {\it \EFX-with-charity}.
Obviously, without any constraints, the problem is trivial: leaving all goods unallocated results in an envy-free allocation.
Thus, the goal here is to find allocations with better guarantees.
For additive valuations, Caragiannis et al.~\cite{caragiannis2019envy} showed that there exists an \EFX allocation with some unallocated goods where every agent receives at least half the value of her bundle in a maximum {\it Nash social welfare} allocation\footnote{This is an allocation that maximizes $\Pi_{i = 1}^n v_i(X_i)$.}.
For normalized and monotone valuations, Chaudhury et al.~\cite{chaudhury2021little} showed that there exist an \EFX allocation and a set of unallocated goods $U$ such that every agent has value for her own bundle at least value for $U$, and $|U| < n$.
Berger et al.~\cite{berger2022almost} showed that the number of the unallocated goods can be decreased to $n-2$, and to just one for the case of four agents having nice cancelable valuations, which are more general than additive valuations.
Mahara~\cite{mahara2021extension} showed that the number of the unallocated goods can be decreased to $n-2$ for normalized and monotone valuations, which are more general than nice cancelable valuations.
For additive valuations, Chaudhury et al.~\cite{chaudhury2021improving} presented a polynomial-time algorithm for finding 
an approximate \EFX allocation with at most a sublinear number of unallocated goods and high Nash social welfare. 

\section{Preliminaries}
\label{sec: pre}
Let $N=\{1,\ldots, n\}$ be a set of $n$ agents and $M$ be a set of $m$ goods. 
We assume that goods are indivisible: a good can not be split among multiple agents.
Each agent $i \in N$ has a non-negative valuation $v_i: 2^M \rightarrow \mathbb{R}_{\ge 0}$, where $2^M$ is the power set of $M$.
We assume that each valuation $v_i$ is {\it additive}, i.e., $v_i(S)=\sum_{g\in S} v_i(\{g\})$ for any $S\subseteq M$.
Note that since valuations are non-negative and additive, they have to be {\it normalized}: $v_i(\emptyset) = 0$ and {\it monotone}: $S\subseteq T$ implies $v_i(S) \le v_i(T)$ for any $S,T \subseteq M$.
For ease of explanation, we normalize the valuations so that $v_i(M)=1$ for all $i\in N$.

To simplify notation, we denote $\{1,\ldots, k\}$ by $[k]$  for any positive integer $k$, write $v_i(g)$ instead of $v_i(\{g\})$ for $g\in M$, and use $S\setminus g$ and $S\cup g$ instead of $S\setminus \{g\}$ and  $S\cup \{g\}$, respectively.

We say that $X=(X_1,X_2,\ldots, X_n)$ is an {\it allocation of $M$ to $N$} if it is a partition of $M$ into $n$ disjoint subsets such that each set is indexed by $i \in N$. Each $X_i$ is the set of goods given to agent $i$, which we call a {\it bundle}. 
It is simply called an {\em allocation to $N$} if $M$ is clear from context. 
For $i\in N$ and $S \subseteq M$, let $m_i(S)$ denote the value of the least valuable good for agent $i$ in $S$, that is, 
$m_i(S)= \min_{g\in S} \{v_i(g)\}$ if $S \neq \emptyset$ and $m_i(\emptyset)=0$.  
For an allocation $X=(X_1,X_2,\ldots, X_n)$ to $N$, we say that an agent $i$ is {\it \PROPavg-satisfied} by $X$ if 
$$
v_i(X_i) +\frac{1}{n-1}\sum_{k\in [n]\setminus i} m_i(X_k) \ge \frac{1}{n}, 
$$
where we recall that $v_i(M)=1$.
In other words, agent $i$ receives a set of goods for which she has value at least $1/n$ fraction of her total value 
minus the average of minimum value of the set of goods any other agent receives.
An allocation $X$ is called {\it \PROPavg} if every agent  $i \in N$ is \PROPavg-satisfied by $X$.

Let  $G=(V, E)$ be a graph. For $S\subseteq V$, let $\Gamma_{G}(S)=\{v\in V\setminus S \mid (s,v)\in E~\text{for some}~s\in S\}$ denote the set of neighbors of $S$ in $G$.
For $v\in V$, let $G-v$ denote the graph obtained from $G$ by deleting $v$.
A {\it perfect matching} in $G$ is a set of pairwise disjoint edges of $G$ covering all the vertices of $G$.


\section{Key Ingredient: \PROPavg-Graph}
\label{sec:propavggraph}

In order to prove Theorem~\ref{thm: main}, 
we give an algorithm for finding a \PROPavg allocation. 
As described in Section~\ref{sec: tech}, our algorithm is a generalization of the cut-and-choose protocol that consists of the following three steps. 
\begin{enumerate}
\item
We partition the goods into $n$ bundles without assigning them to agents. 
\item 
A specified agent, say $n$, chooses the best bundle for her valuation.  
\item
We determine an assignment of the remaining bundles to the agents in $N \setminus n$. 
\end{enumerate}
The partition given in the first step is represented by an allocation of $M$ to a newly introduced set of size $n$, say $V_2$, and 
the assignment in the third step is represented by a matching in an auxiliary bipartite graph, which we call {\it \PROPavg-graph}. 
In this section, we define the \PROPavg-graph and its desired properties. 

Let $V_2$ be a set of $n$ elements and fix a specified element $r \in V_2$. 
We say that $X=(X_u)_{u\in V_2}$ is an {\it allocation to $V_2$} if it is a partition of $M$ into $n$ disjoint subsets such that each set is indexed by an element in $V_2$, that is,  
$\bigcup_{u \in V_2} X_u = M$ and $X_u \cap X_{u'} = \emptyset$ for distinct $u, u' \in V_2$. 
For an allocation $X=(X_u)_{u\in V_2}$ to $V_2$, we define a bipartite graph $G_{X}=(V_1, V_2; E)$ called {\it \PROPavg-graph} as follows. 
The vertex set consists of $V_1 = N \setminus n$ and $V_2$, and the edge set $E$ is defined by
$$
(i,u) \in E \iff v_i(X_u)+\frac{1}{n-1}\sum_{u'\in V_2 \setminus \{r, u\}} m_i(X_{u'}) \ge \frac{1}{n}
$$
for $i \in V_1$ and $u \in V_2$. 
It should be emphasized that the summation is taken over $V_2 \setminus \{r, u\}$, i.e., $m_i(X_{r})$ is not  counted, in the above definition, which is crucial in our argument. 
The following lemma shows that the \PROPavg-graph is closely related to the definition of \PROPavg-satisfaction. 

\begin{lemma}\label{ob: v1}
Suppose that $G_X=(V_1, V_2; E)$ is the \PROPavg-graph for  an allocation $X=(X_u)_{u\in V_2}$ to $V_2$. 
Let $\sigma$ be a bijection from $N$ to $V_2$  and 
define an allocation $Y=(Y_1,\ldots, Y_n)$ to $N$ by $Y_i = X_{\sigma(i)}$ for $i \in N$. 
For $i^*\in V_1$, if $(i^*, \sigma(i^*)) \in E$, then 
$i^*$ is \PROPavg-satisfied by $Y$.
\end{lemma}
\begin{proof}
Let $u^* = \sigma(i^*)$ and suppose that $(i^*, u^*) \in E$. 
We directly obtain
\begin{equation*}
v_{i^*}(Y_{i^*})+\frac{1}{n-1}\sum_{j\in [n]\setminus {i^*}} m_{i^*}(Y_j)\ge v_{i^*}(X_{u^*})+\frac{1}{n-1}\sum_{u\in V_2 \setminus \{u^*, r\}} m_{i^*}(X_u) \ge \frac{1}{n},
\end{equation*}
where the first inequality follows from the definition of $Y$ and $m_{i^*}(X_r) \ge 0$, and the second inequality follows from $(i^*, u^*)\in E$.
\end{proof}
\begin{figure}[tbp]
  \begin{center}
   \includegraphics[scale=0.5]{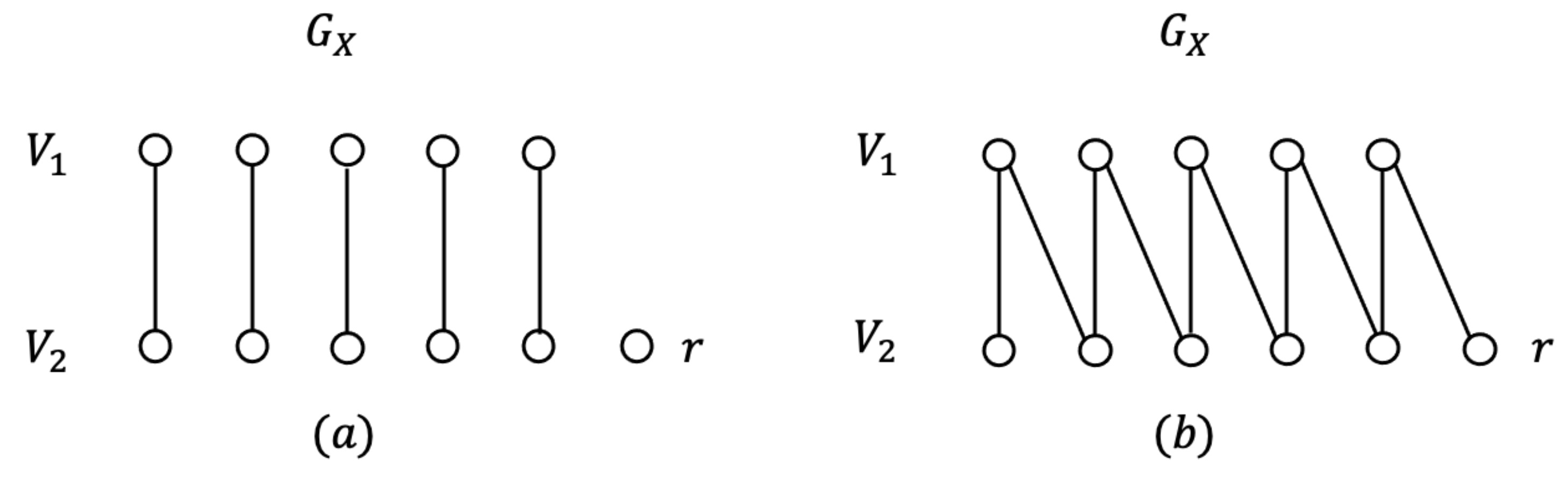}
 \end{center}
  \caption{An example of $(a)$ represents a \PROPavg-graph $G_X$ corresponding to $X$ that satisfies (P1) but does not satisfy (P2). An example of $(b)$ represents a \PROPavg-graph $G_X$ corresponding to $X$ that satisfies (P2).}
  \label{fig: 2}
\end{figure}

For an allocation $X=(X_u)_{u\in V_2}$ to $V_2$, 
we consider the following two properties, which play an important role in our argument.
\begin{description}
\item[(P1)] $G_X- r$ has a perfect matching.
\item[(P2)] For any $u \in V_2$, $G_X-u$ has a perfect matching.
\end{description}
Obviously, property (P2) is stronger than property (P1). 
Roughly speaking, our goal is to find an allocation $X=(X_u)_{u\in V_2}$ to $V_2$ satisfying (P2).
As shown later, we find an allocation $X=(X_u)_{u\in V_2}$ to $V_2$ satisfying (P2) while keeping (P1).
Examples of a \PROPavg-graph $G_X$ corresponding to $X$ satisfying (P1) or (P2) are described in Figure~\ref{fig: 2}.
We can rephrase these conditions by using the following classical theorem known as Hall's marriage theorem in discrete mathematics.
\begin{theorem}[Hall's marriage theorem~\cite{hall1935representatives}]
Suppose that $G=(A, B; E)$ is a bipartite graph with $|A|=|B|$. 
Then, $G$ has a perfect matching if and only if $|S| \le |\Gamma_G(S)|$ for any $S\subseteq A$.
\end{theorem}

The property (P1) is equivalent to $|S| \le |\Gamma_{G_{X}- r}(S)|$ for any $S\subseteq V_1$ by this theorem.
The property (P2) is equivalent to $|S| \le |\Gamma_{G_{X}- u}(S)|$ for any $u\in V_2$ and $S\subseteq V_1$ by Hall's marriage theorem.
By simple observation, we can obtain another characterization of property (P2).
\begin{lemma}\label{ob: p2}
Let $X=(X_u)_{u\in V_2}$ be an allocation to $V_2$.
Then, 
$X$ satisfies (P2) if and only if $|S|+1 \le |\Gamma_{G_{X}}(S)|$ for any non-empty subset $S\subseteq V_1$.
\end{lemma}
\begin{proof}
By Hall's marriage theorem, it is sufficient to show that the following two conditions are  equivalent: 
\begin{enumerate}
\item[(i)] $|S| \le |\Gamma_{G_{X}- u}(S)|$ for any $u\in V_2$ and $S\subseteq V_1$, and
\item[(ii)] $|S|+1 \le |\Gamma_{G_{X}}(S)|$ for any non-empty subset $S\subseteq V_1$.
\end{enumerate}

Suppose that (i) holds. Let $S$ be a nonempty subset of $V_1$. 
Since (i) implies that $|\Gamma_{G_{X}}(S)| \ge |S| \ge 1$, we obtain $\Gamma_{G_{X}}(S) \not= \emptyset$. 
Let $u \in \Gamma_{G_{X}}(S)$. 
By (i) again, we obtain $|\Gamma_{G_{X}}(S)| = |\Gamma_{G_{X} -u}(S)| + 1 \ge |S| + 1$. 
This shows (ii). 

Conversely, suppose that (ii) holds. Let $u \in V_2$ and let $S \subseteq V_1$. 
If $S=\emptyset$, then it clearly holds that $|S| \le |\Gamma_{G_{X}- u}(S)|$.
If $S \neq \emptyset$, then 
we have $|S|+1 \le |\Gamma_{G_{X}}(S)|\le|\Gamma_{G_X-u}(S)|+1$, which implies that $|S| \le |\Gamma_{G_X-u}(S)|$.
This shows (i). 
\end{proof}

\section{Existence of a \PROPavg Allocation}
\label{sec:mainproof}
In this section, we show that there always exists a \PROPavg allocation by constructing a pseudo-polynomial time algorithm for finding it.
Improvements in time complexity are discussed in Section~\ref{sec:poly}.
Our algorithm begins with obtaining an initial allocation ${X}=(X_u)_{u \in V_2}$ to $V_2$ satisfying (P1). 
Unless $X$ satisfies (P2), we appropriately choose a good in $\bigcup_{u \in V_2 \setminus r} X_u$ and move it to $X_r$ while keeping  (P1).
Finally, we get an allocation ${X^*}=(X^*_u)_{u \in V_2}$ to $V_2$ satisfying (P2).
As we will see later, we can obtain a \PROPavg allocation to $N$ from $X^*$. 

\subsection{Our Algorithm}

In order to obtain an initial allocation ${X}=(X_u)_{u \in V_2}$ to $V_2$ satisfying (P1), we use the following previous result about {\it \EFX-with-charity}.
\label{sec: propavg}
\begin{theorem} [Chaudhury et al.~\cite{chaudhury2021little}]\label{thm: charity}
For normalized and monotone valuations, there always exists an allocation $X=(X_1,\ldots, X_n)$ of $M \setminus U$ to $N$, where $U$ is a set of unallocated goods, such that 
\begin{itemize}
\item X is \EFX, that is, $v_i(X_i)+m_i(X_j) \ge v_i(X_j)$ for any pair of agents $i,j \in N$, 
\item $v_i(X_i) \ge v_i(U)$ for any agent $i \in N$, and 
\item $|U| < n$.
\end{itemize}
\end{theorem}
The following lemma shows that by applying Theorem~\ref{thm: charity} to agents $N \setminus n$, we can obtain an initial allocation ${X}=(X_u)_{u \in V_2}$ to $V_2$ satisfying (P1).
\begin{lemma}\label{lem: initial}
There exists an allocation ${X}=(X_u)_{u \in V_2}$ to $V_2$ satisfying (P1).
\end{lemma}
\begin{proof}
By applying Theorem~\ref{thm: charity} to agents  $N \setminus n$, we can obtain an allocation $Y =(Y_1,\ldots, Y_{n-1})$ of $M \setminus U$ to $N \setminus n$, where $U$ is a set of unallocated goods, 
satisfying the conditions in Theorem~\ref{thm: charity}.
Let $V_2 = \{r, u_1, \dots , u_{n-1}\}$ and 
define  an allocation $X =(X_u)_{u \in V_2}$ to $V_2$ as $X_{u_j} =Y_j$ for $j \in [n-1]$ and $X_r = U$. 
Let $G_X=(V_1, V_2; E)$ be the \PROPavg-graph for $X$.
We show that $X$ satisfies (P1).

Fix any agent $i \in V_1$.
We have $v_i(X_{u_i})+m_i(X_{u_j}) \ge v_i(X_{u_j})$ for any $j\in [n-1]\setminus i$ since $Y$ is \EFX and $X_{u_j} =Y_j$.
We also have $v_i(X_{u_i}) = v_i(Y_i) \ge v_i(U) = v_i(X_r)$ and a trivial inequality $v_i(X_{u_i}) \ge v _i(X_{u_i})$.
By summing up these inequalities, we obtain $n \cdot v_i(X_{u_i}) + \sum_{j\in [n-1]\setminus i} m_i(X_{u_j}) \ge \sum_{u \in V_2} v_i (X_u)$. 
Since $\sum_{u \in V_2} v_i (X_u) = v_i (M) = 1$,  this shows that $v_i(X_{u_i})+\frac{1}{n}\sum_{j\in [n-1]\setminus i} m_i(X_{u_j}) \ge \frac{1}{n}$, and hence $(i, u_i)\in E$.
Therefore,  $G_{X} - r$ has a perfect matching $\{ (i, u_i) \mid i \in [n-1]\}$, which implies that $X$ satisfies (P1). 
\end{proof}
The following lemma shows that if we obtain an allocation ${X}=(X_u)_{u \in V_2}$ to $V_2$ satisfying (P2), then there exists a \PROPavg allocation to $N$. 
\begin{lemma}\label{lem: termination}
Suppose that ${X}=(X_u)_{u \in V_2}$ is an allocation to $V_2$ satisfying (P2). 
Then, we can construct a \PROPavg allocation to $N$.
\end{lemma}
\begin{proof}
Let ${X}=(X_u)_{u \in V_2}$ be an allocation to $V_2$ satisfying (P2).
First, agent $n$ chooses the best bundle $X_{u^*}$ for her valuation among $\{X_u \mid u \in V_2 \}$ (if there is more than one such bundle, choose one arbitrarily).
Since $X$ satisfies (P2), there exists a perfect matching $A$ in $G_X - {u^*}$. 
For each agent $i\in V_1 (= N \setminus n)$, the bundle that matches $i$ in $A$ is allocated to $i$.
By Lemma~\ref{ob: v1}, agent $i$ is \PROPavg-satisfied for each agent $i\in V_1$.
Furthermore, since we have $v_n(X_{u^*}) = \max_{u\in V_2} v_n(X_u) \ge \frac{1}{n}$, agent $n$ is also \PROPavg-satisfied.
Therefore, the obtained allocation is a \PROPavg allocation to $N$.
\end{proof}

The following proposition shows how we update an allocation in each iteration, whose proof is given in Section~\ref{sec: update}. 

\begin{proposition}\label{thm: update}
Suppose that $X=(X_u)_{u \in V_2}$ is an allocation to $V_2$ that satisfies (P1) but does not satisfy (P2).
Then, there exists another allocation $X'=(X'_u)_{u \in V_2}$ to $V_2$ satisfying (P1) such that $|X'_{r}|=|X_{r}|+1$. 
\end{proposition}
As we will see in Section~\ref{sec: update}, an allocation $X'$ in Proposition~\ref{thm: update} is obtained by moving some appropriate item $g \in \bigcup_{u \in V_2 \setminus r} X_u$ to $X_r$.
If Proposition~\ref{thm: update} holds, then we can show that there always exists a \PROPavg allocation when each agent has a non-negative additive valuation as follows. See also Algorithm~\ref{alg:01}. 

By Lemma~\ref{lem: initial}, we first obtain an initial allocation $X=(X_u)_{u \in V_2}$ to $V_2$ satisfying (P1). 
By Proposition~\ref{thm: update}, unless $X$ satisfies (P2), we can increase $|X_r|$ by one while keeping the condition (P1).
Since $|X_r|\le |M|$, this procedure terminates in at most $m$ steps, and we finally obtain an allocation ${X}^*$ to $V_2$ satisfying (P2).
Therefore, there exists a \PROPavg allocation to $N$ by Lemma~\ref{lem: termination}.

\begin{algorithm}[htb]
\caption{ Algorithm for finding a \PROPavg allocation}
\label{alg:01}
\begin{algorithmic}[1]
\Require agents $N$, goods $M$, and valuation $v_i$ for each $i \in N$
\Ensure  a \PROPavg allocation to $N$
\State Apply Lemma~\ref{lem: initial} to obtain an allocation $X$ to $V_2$ satisfying (P1). 
\While{$X$ does not satisfy (P2)}
    \State Apply Proposition~\ref{thm: update} to $X$ and obtain another allocation $X'$ to $V_2$. 
    \State $X \leftarrow X'$.
\EndWhile
\State Apply Lemma~\ref{lem: termination} to obtain a \PROPavg allocation to $N$. 
\end{algorithmic}
\end{algorithm}

\subsection{Proof of Proposition~\ref{thm: update}}
\label{sec: update}


Let $X=(X_u)_{u \in V_2}$ be an allocation to $V_2$. 
For $u^* \in V_2\setminus r$ and $g \in X_{u^*}$, 
we say that an allocation $X'=(X'_u)_{u \in V_2}$ to $V_2$ is obtained from $X$ by {\it moving $g$ to $X_r$} 
if \begin{align*}
X'_u &= \left\{ 
\begin{array}{ll}
X_r \cup g & {\rm if}~u=r, \\
X_{u^*}\setminus g & {\rm if}~u=u^*,\\
X_{u} & {\rm otherwise}.
\end{array}
 \right. &
\end{align*}

The following lemma guarantees that if there exists an agent $i \in V_1$ such that $(i,r)\not \in E$ in the \PROPavg-graph $G_X = (V_1, V_2; E)$, then we can move some good in $\bigcup_{u \in V_2 \setminus r} X_u$ to $X_r$ so that the edges incident to $i$ do not disappear.
This lemma is crucial in the proof of Proposition~\ref{thm: update}.

\begin{lemma}\label{lem: removable}
Let $X=(X_u)_{u \in V_2}$ be an allocation to $V_2$ and let $i \in V_1$ be an agent such that $(i,r)\not \in E$ in the \PROPavg-graph $G_X = (V_1, V_2; E)$. 
Then, there exist $u^*\in V_2$ and $g\in X_{u^*}$ such that 
$(i, u^*)\in E$, $|X_{u^*}|\ge 2$, and the following property holds: 
if an allocation $X'$ to $V_2$ is obtained from $X$ by moving $g$ to $X_r$, then 
the corresponding \PROPavg-graph $G_{X'}$ has an edge $(i, u^*)$. 
\end{lemma}
\begin{proof}
We first show that $X_u \neq\emptyset$ for any $u\in V_2$ with $(i,u)\in E$.
Indeed, if $X_u = \emptyset$, then we have 
$$
v_i(X_r)+\frac{1}{n-1}\sum_{u' \in V_2\setminus r} m_i(X_{u'})\ge v_i(X_u) + \frac{1}{n-1}\sum_{u' \in V_2\setminus \{r, u\}} m_i(X_{u'})\ge \frac{1}{n}, 
$$
where
the first inequality follows from $v_i(X_u)=m_i(X_u)=0$ and the second inequality follows from $(i, u) \in E$.
This contradicts $(i, r)\not \in E$.

To derive a contradiction, assume that 
$u^*$ and $g$ satisfying the conditions in Lemma~\ref{lem: removable} do not exist. 
Then, we have the following claim. 

\begin{claim}\label{cl: ineq}
For any $u\in V_2$ with $(i, u)\in E$, we obtain 
\begin{equation}
v_i(X_u)-m_i(X_u)+\frac{1}{n-1}\sum_{\substack{u'\in V_2\setminus \{r, u\}:\\(i, u')\in E}} m_i(X_{u'}) <\frac{1}{n}. \label{eq: 01}
\end{equation}
\end{claim}
\begin{proof}[Proof of the Claim]
Fix $u\in V_2$ with $(i, u)\in E$. 
Let $g$ be a good in $X_u$ that minimizes $v_i(g)$, where we note that $X_u \not= \emptyset$ as described above. Then, $v_i(g) = m_i(X_u)$. 
Define $X'=(X'_{u'})_{u' \in V_2}$ as the allocation to $V_2$ that is obtained from $X$ by moving $g$ to $X_r$. 
Let $G_{X'}=(V_1, V_2; E')$ be the \PROPavg-graph corresponding to $X'$.
Since $u$ and $g$ do not satisfy the conditions in Lemma~\ref{lem: removable} by our assumption, we have $(i, u) \not\in E'$ or $|X_u| =1$. 

If  $(i, u)\in E'$, then we obtain  $|X_u| = 1$, and hence 
\begin{align*}
v_i(X_{r})+\frac{1}{n-1}\sum_{u'\in V_2\setminus {r}} m_i(X_{u'})  
&\ge \frac{1}{n-1}\sum_{u'\in V_2 \setminus \{r, u\}} m_i(X_{u'})  \\
&= v_i(X'_{u})+\frac{1}{n-1}\sum_{u'\in V_2\setminus \{r, u\}} m_i(X'_{u'}) \\
&\ge \frac{1}{n}, 
\end{align*}
where 
the equality follows from $v_i(X'_{u}) = v_i(\emptyset) = 0$
and the last inequality follows from $(i, u)\in E'$.
This contradicts $(i, r)\not \in E$.

Thus, it holds that $(i, u) \not\in E'$.
Since $v_i(X_u)-m_i(X_u) = v_i(X'_u)$, we obtain 
\begin{align*}
&v_i(X_u)-m_i(X_u)+\frac{1}{n-1}\sum_{\substack{u'\in V_2\setminus \{r, u\}:\\(i, u')\in E}} m_i(X_{u'})  \\
&\le v_i(X_u)-m_i(X_u)+\frac{1}{n-1}\sum_{u'\in V_2\setminus \{r, u\}} m_i(X_{u'}) \\
&= v_i(X'_u)+\frac{1}{n-1}\sum_{u'\in V_2\setminus \{r, u\}} m_i(X'_{u'}) \\
&< \frac{1}{n}, 
\end{align*}
where the last inequality follows from $(i, u) \not\in E'$.  
\end{proof}
By summing up inequality (\ref{eq: 01}) for each $u\in V_2$ with $(i,u)\in E$, we obtain the following inequality: 
\begin{equation}\label{eq: b}
\sum_{\substack{u\in V_2:\\ (i,u)\in E}} v_i(X_{u})+\left(- 1 + \frac{l-1}{n-1}\right)\sum_{\substack{u'\in V_2\setminus r: \\ (i, u') \in E}} m_i(X_{u'}) < \frac{l}{n},
\end{equation}
where $l=|\{u\in V_2\mid (i, u)\in E\}|$. 

On the other hand, for any $u\in V_2$ with $(i,u)\not\in E$, we have 
\begin{equation}\label{eq: d}
v_i(X_{u})+\frac{1}{n-1}\sum_{\substack{u'\in V_2\setminus r: \\ (i,u')\in E}} m_i(X_{u'}) \le 
v_i(X_{u})+\frac{1}{n-1}\sum_{u'\in V_2\setminus \{r, u\}} m_i(X_{u'}) < \frac{1}{n}, 
\end{equation}
where the both inequalities follow from $(i,u)\not\in E$.
Summing up inequality~(\ref{eq: d}) for each $u\in V_2$ with $(i,u)\not\in E$, we obtain 
\begin{equation}\label{eq: e}
\sum_{\substack{u\in V_2:\\ (i,u)\not\in E}} v_i(X_u)+\left(\frac{n-l}{n-1}\right)\sum_{\substack{u'\in V_2\setminus r: \\ (i,u') \in E}} m_i(X_{u'}) < \frac{n-l}{n}, 
\end{equation}
where we note that $|\{u\in V_2\mid (i, u) \not\in E\}| = n-l$. 

By taking the sum of inequalities (\ref{eq: b}) and (\ref{eq: e}), we obtain 
$$
\sum_{\substack{u\in V_2:\\ (i,u)\in E}} v_i(X_{u})+ \sum_{\substack{u\in V_2:\\ (i,u)\not\in E}} v_i(X_{u}) < 1, 
$$
which contradicts $\sum_{u\in V_2} v_i(X_{u}) = 1$. 

Therefore, there exist $u^*\in V_2$ and $g\in X_{u^*}$ satisfying the conditions in Lemma~\ref{lem: removable}.
\end{proof}
We are now ready to prove Proposition~\ref{thm: update}.
\begin{proof}[Proof of Proposition~\ref{thm: update}]
Suppose that $X=(X_u)_{u \in V_2}$ is an allocation to $V_2$ that satisfies (P1) but does not satisfy (P2).
Let $G_X = (V_1, V_2; E)$ be the \PROPavg-graph corresponding to $X$.
Since $X$ does not satisfy (P2), there exists a non-empty set $S \subseteq V_1$ such that $|S|+1 > |\Gamma_{G_X}(S)|$ by Lemma~\ref{ob: p2}. 
Among such sets, let $S^* \subseteq V_1$ be an inclusion-wise minimal one. 
By the integrality of  $|\Gamma_{G_X}(S^*)|$ and $|S^*|$, this means that
$|S^*| \ge |\Gamma_{G_X}(S^*)|$ and $|S|+1 \le |\Gamma_{G_X}(S)|$ for any non-empty proper subset $S \subsetneq S^*$. 
We now show some properties of $S^*$. 

\begin{claim}\label{clm: S*01}
For any $i\in S^*$, it holds that $(i, r) \not\in E$. 
\end{claim}

\begin{proof}[Proof of the claim]
Since $X$ satisfies (P1), we have $|S^*| \le |\Gamma_{G_X-r}(S^*)|$ by Hall's marriage theorem.
Hence, we obtain $|S^*| \le |\Gamma_{G_X-r}(S^*)| \le |\Gamma_{G_X}(S^*)| \le |S^*|$, where the last inequality follows from the definition of $S^*$. 
This shows that all the above inequalities are tight. 
Since $|\Gamma_{G_X-r}(S^*)| = |\Gamma_{G_X}(S^*)|$, we obtain $r \not\in \Gamma_{G_X}(S^*)$, that is,  $(i, r) \not\in E$ for any $i\in S^*$.
\end{proof}

\begin{claim}\label{clm: S*02}
For any non-empty proper subset $S \subsetneq S^*$, it holds that $|S|+1\le |\Gamma_{G_X -r} (S)|$.
\end{claim}

\begin{proof}[Proof of the claim]
Let $S$ be a non-empty proper subset of $S^*$. 
Since $(i, r) \not\in E$ for any $i\in S \subsetneq S^*$ by Claim~\ref{clm: S*01}, 
we have that $\Gamma_{G_X} (S) = \Gamma_{G_X-r} (S)$. 
Hence, we obtain $|S|+1\le |\Gamma_{G_X} (S)| = |\Gamma_{G_X -r} (S)|$, 
where the inequality follows from the minimality of $S^*$.
\end{proof}

\begin{claim}\label{cl: matching}
For any $i\in S^*$ and $u\in \Gamma_{G_X}(S^*)$ with $(i, u)\in E$, $G_X-r$ has a perfect matching in which $i$ matches $u$.
\end{claim}
\begin{proof}[Proof of the claim]
Fix any $i\in S^*$ and $u\in \Gamma_{G_X}(S^*)$ with $(i, u)\in E$.
Note that $r \not\in \Gamma_{G_X}(S^*)$ by Claim~\ref{clm: S*01}, and hence $u \neq r$. 

Since $X$ satisfies (P1), $G_X - r$ has a perfect matching $A$. 
In $A$, it is obvious that every vertex in $S^*$ is matched to a vertex in $\Gamma_{G_X -r}(S^*)$.  
Conversely, every vertex in $\Gamma_{G_X -r}(S^*)$ is matched to a vertex in $S^*$ as $|S^*| = |\Gamma_{G_X -r}(S^*)|$ (see the proof of Claim~\ref{clm: S*01}). 
Thus, by removing the edges between $S^*$ and $\Gamma_{G_X}(S^*)$ from $A$, 
we obtain a matching  $A_1 \subseteq A$ that exactly covers  $V_1\setminus S^*$ and  $V_2\setminus (\Gamma_{G_X}(S^*)\cup \{r\})$.

Let $G'_X$ be the subgraph of $G_X$ induced by $(S^*\setminus i)\cup (\Gamma_{G_X}(S^*)\setminus u)$.
We now show that $G'_X$ has a perfect matching. 
Consider any $S\subseteq S^*\setminus i$.
If $S=\emptyset$, then it clearly holds that $|S|\le |\Gamma_{G'_{X}}(S)|$.
If $S\neq \emptyset$, then $|S|+1\le |\Gamma_{G_X-r} (S)| \le  |\Gamma_{G'_X} (S) \cup u| =  |\Gamma_{G'_X} (S)|+1$, where the first inequality is by Claim~\ref{clm: S*02}.  
Therefore, $|S|\le |\Gamma_{G'_{X}}(S)|$ holds for any $S\subseteq S^*\setminus i$, and hence 
$G'_X$ has a perfect matching $A_2$ by Hall's marriage theorem.

Then, $A_1\cup A_2 \cup \{(i, u)\}$ is a desired perfect matching in $G_X-r$. 
\end{proof}
\begin{figure}[tbp]
  \begin{center}
   \includegraphics[scale=0.5]{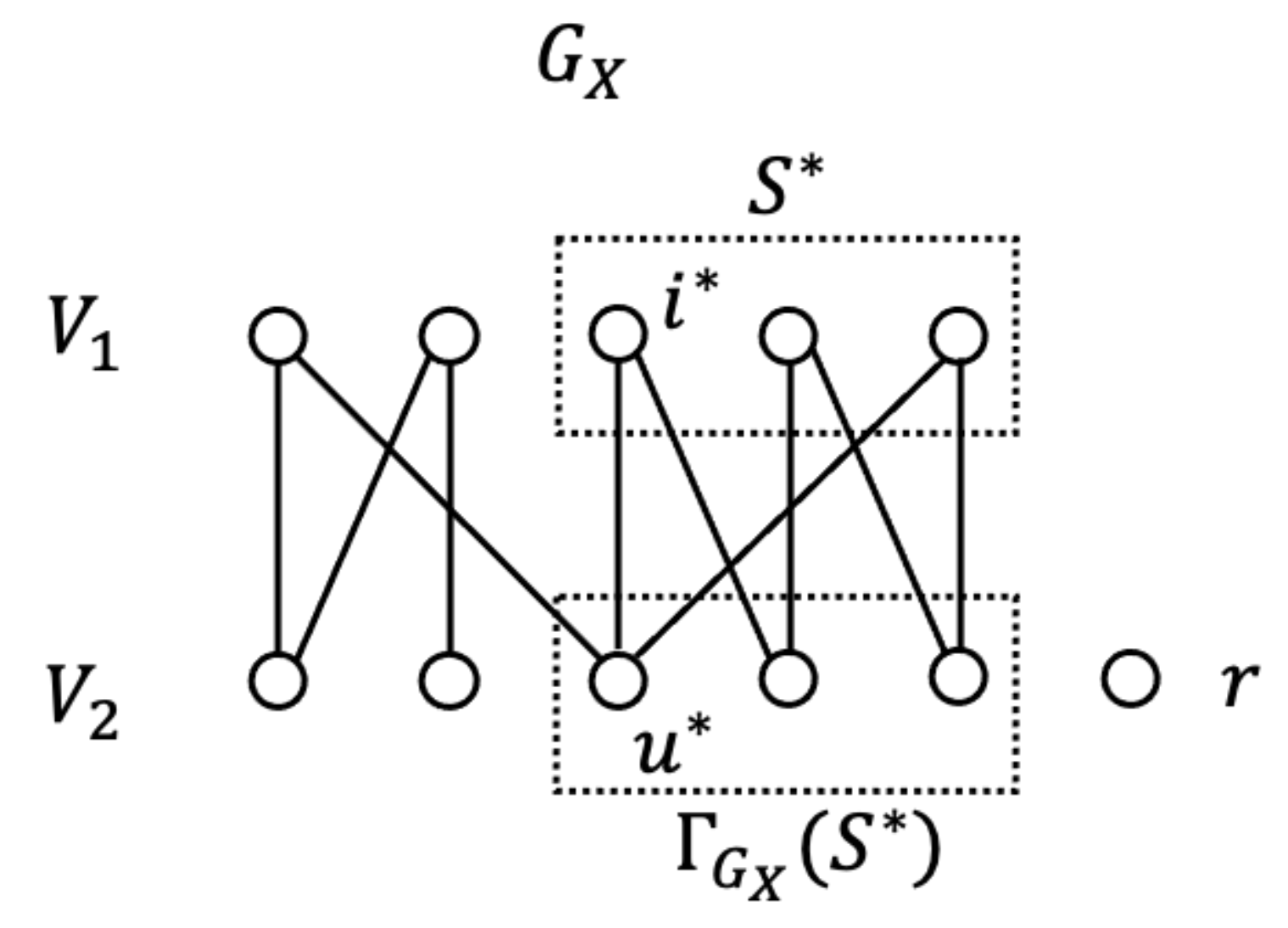}
 \end{center}
  \caption{A \PROPavg-graph $G_X$ corresponding to $X$ in the proof of Proposition~\ref{thm: update}. A non-empty minimal subset $S^*$ of $V_1$ such that $|S^*|+1 > |\Gamma_{G_X}(S^*)|$ is represented by vertices in the above dotted rectangle.}
  \label{fig: 3}
\end{figure}

Fix any agent $i^* \in S^*$.
Since $(i^*,r)\not\in E$ by Claim~\ref{clm: S*01}, by applying Lemma~\ref{lem: removable} to agent $i^*$, 
we obtain $u^*\in V_2$ and $g\in X_{u^*}$ satisfying the conditions in Lemma~\ref{lem: removable}.
See also Figure~\ref{fig: 3}.
Let $X'=(X'_u)_{u\in V_2}$ be the allocation to $V_2$ obtained from $X$ by moving $g$ to $X_r$ and let $G_{X'}=(V_1, V_2; E')$ be the \PROPavg-graph corresponding to $X'$.
Then, the conditions in Lemma~\ref{lem: removable} show that $(i^*, u^*)\in E \cap E'$ and $|X_{u^*}|\ge 2$. 
We also see that $E'$ satisfies the following. 

\begin{claim}\label{cl: keepingedges}
For any $i\in V_1$ and $u\in V_2\setminus u^*$, if $(i, u) \in E$ then $(i, u)\in E'$. 
\end{claim}
\begin{proof}[Proof of the claim]
Since $|X_{u^*}|\ge 2$, we have $m_i(X'_{u^*})=m_i(X_{u^*}\setminus g)\ge m_i(X_{u^*})$ for any agent $i\in V_1$.
Hence, for any $i\in V_1$ and $u\in V_2\setminus u^*$ with $(i, u) \in E$, we obtain 
$$
v_i(X'_u) + \frac{1}{n-1}\sum_{u'\in V_2 \setminus \{r, u\}} m_i(X'_{u'}) \ge
v_i(X_u) + \frac{1}{n-1}\sum_{u'\in V_2 \setminus \{r, u\}} m_i(X_{u'}) \ge \frac{1}{n}, 
$$
which shows that $(i, u)\in E'$. 
\end{proof}

By Claim~\ref{cl: matching} and $(i^*, u^*)\in E$, there exists a perfect matching $A$ in $G_X-r$ in which $i^*$ matches $u^*$.
Then, by Claim~\ref{cl: keepingedges} and $(i^*, u^*)\in E'$, we see that $A \subseteq E'$, that is, $A$ is a perfect matching also in $G_{X'}-r$.
Therefore, $X'$ satisfies (P1).
Since $|X'_r|=|X_r|+1$ clearly holds by definition, 
$X'$ satisfies the conditions in Proposition~\ref{thm: update}. 
\end{proof}

\section{Finding a \PROPavg Allocation in Polynomial Time}\label{sec:poly}
In this section, we show how to find a \PROPavg allocation in polynomial time.
As mentioned in Section~\ref{sec:mainproof}, Algorithm~\ref{alg:01} runs in pseudo-polynomial time.
This is because we can not guarantee the polynomial solvability in line 1 of Algorithm~\ref{alg:01}.
We can see that the other parts of Algorithm~\ref{alg:01} run in polynomial time as follows.
In line 2, we can check (P2) in polynomial time by applying a maximum
matching algorithm for each $G_X - u$.
In line 3, it suffices to find a good $g \in \bigcup_{u \in V_2 \setminus
r} X_u$ such that (P1) is kept after moving $g$.
Since (P1) can be checked in polynomial time, 
this can be done in polynomial time by considering all $g$ in a
brute-force way.
Finally, line 5 is executed in polynomial time by Lemma~\ref{lem: termination}.
Note that we can speed up lines 2 and 3 by using the
DM-decomposition of $G_X$ \cite{dulmage1958coverings,
dulmage1959structure},
but we do not go into details, because we only focus on the polynomial solvability.

Let us now consider how to find an initial allocation $X$ to $V_2$ satisfying (P1) in polynomial time.
Our idea is to use a recursive algorithm.
That is, we use a \PROPavg allocation of $M$ to $n-1$ agents as an initial allocation $X$ to $V_2$ satisfying (P1).
Indeed, if it holds that $v_i(g)\le \frac{1}{n}$ for any agent $i\in N$ and any good $g\in M$, then we can show that a \PROPavg allocation of $M$ to $n-1$ agents satisfies (P1) as follows.

\begin{lemma}\label{lem: rec}
Suppose that for any agent $i\in N$ and any good $g\in M$, we have $v_i(g)\le \frac{1}{n}$.
Let $(X_1,\ldots ,X_{n-1})$ be a \PROPavg allocation for $N\setminus n$.
Then, $X=(X_1,\ldots ,X_{n-1}, X_n )$ is an allocation to $V_2=[n]$ satisfying (P1), where $X_n = \emptyset$ and the specific element $r\in V_2$ is equal to $n$.
\end{lemma}
\begin{proof}
Let $G_X=(V_1, V_2; E)$ be the \PROPavg-graph corresponding to $X$.
It is enough to show that $(i, X_i)\in E$ for any $i\in [n-1]$.
Fix any $i\in [n-1]$.
We obtain that 
\begin{align*}
v_i(X_i) &\ge \frac{1}{n-1}-\frac{1}{n-2}\sum_{j\in [n-1]\setminus i} m_i(X_j)\\
&=\frac{1}{n}-\frac{1}{n-1}\sum_{j\in [n-1]\setminus i} m_i(X_j)\\
&+\underbrace{ \frac{1}{n-1}\left(\frac{1}{n} -\frac{1}{n-2}\sum_{j\in [n-1]\setminus i} m_i(X_j)\right)}_{\ge 0}\\
&\ge \frac{1}{n}-\frac{1}{n-1}\sum_{j\in [n-1]\setminus i} m_i(X_j), 
\end{align*}
where the first inequality follows from the assumption that $(X_1,\ldots ,X_{n-1})$ is a \PROPavg allocation and the second inequality follows from the assumption that $v_i(g) \le \frac{1}{n}$ for any $i\in N$ and $g\in M$.
This implies that $(i, X_i)\in E$ and thus $X$ is an allocation to $V_2=[n]$ satisfying (P1).
\end{proof}

Unfortunately, the argument in Lemma~\ref{lem: rec} does not work without the assumption that $v_i(g)\le \frac{1}{n}$ for any $i\in N$ and $g\in M$.
To elude this difficulty, our algorithm applies preprocessing.
This preprocessing allocates $g$ to $i$ and remove $i$ and $g$ from our instance as long as there exists an agent $i$ and a good $g$ such that $v_i(g)\ge \frac{1}{n}$.
See Algorithm~\ref{alg:02} for the entire algorithm.

If this preprocessing removes at least one agent from our instance, then our algorithm recursively computes a \PROPavg allocation for the remaining agents and goods, and returns the overall allocation together with the removed agents.
In order to verify that the returned allocation is a \PROPavg allocation for $n$ agents, we need a refined condition (see line 7 of Algorithm~\ref{alg:02}).

Otherwise, our algorithm recursively computes a \PROPavg allocation for $n-1$ agents.
Since $v_i(g)< \frac{1}{n}$ holds for any agent $i$ and good $g$,  
we can use this allocation as an initial allocation to $V_2$ satisfying (P1) by Lemma~\ref{lem: rec}.
The rest of our algorithm finds an allocation to $V_2$ satisfying (P2) and returns a \PROPavg allocation as in Algorithm~\ref{alg:01}.

In the remaining part of this section, we show the correctness and the polynomial solvability of Algorithm~\ref{alg:02} .
The following lemma shows that if the preprocessing removes at least one agent from our instance, then the algorithm returns a legal \PROPavg allocation for $N$.
\begin{algorithm}[htb]
\caption{ Algorithm for finding a \PROPavg allocation in polynomial time}
\label{alg:02}
\begin{algorithmic}[1]
\Procedure{PROPavg}{$N$, $M$, $\{v_i\}_{i \in N}$}
\If{$|N|=1$}
\State \Return $X=(M)$
\Else
\State $N_1 \leftarrow N, N_2 \leftarrow \emptyset$
\State $M_1 \leftarrow M, M_2 \leftarrow \emptyset$
\While{$\exists i \in N_1$ and $\exists g\in M_1$ s.t. $v_i(g) \ge \frac{v_i(M)}{|N|}-\frac{1}{|N|-1}\sum_{j\in N_2} m_i(X_j)$}
\State $X_i \leftarrow \{g\}$
\State $N_1\leftarrow N_1\setminus i, N_2 \leftarrow N_2\cup i$
\State $M_1\leftarrow M_1\setminus g, M_2 \leftarrow M_2\cup g$
\EndWhile
\State Let $N_1=\{1,\ldots, l\}$ and $N_2=\{l+1,\ldots, |N|\}$, renumbering if necessary.
\If{$|N_2|\ge 1$}
\State $(X_1,\ldots, X_l) \leftarrow$ { \sc PROPavg}($N_1$, $M_1$, $\{v_i\}_{i \in N_1}$)
\State \Return $X=(X_1,\ldots, X_{|N|})$
\Else
\State $(X_1,\ldots, X_{n-1}) \leftarrow$ { \sc PROPavg}($N\setminus n$, $M$, $\{v_i\}_{i \in N\setminus n}$)
\Comment{$N_1=N, M_1=M$}
\State Apply Lemma~\ref{lem: rec} to obtain an allocation $X=(X_1,\ldots, X_{n})$ satisfying (P1). 
\While{$X$ does not satisfy (P2)}
    \State Apply Proposition~\ref{thm: update} to $X$ and obtain another allocation $X'$ to $V_2$. 
    \State $X \leftarrow X'$.
\EndWhile
\State Apply Lemma~\ref{lem: termination} to obtain a \PROPavg allocation $X=(X_1,\ldots, X_{|N|})$ to $N$.
\State \Return $X=(X_1,\ldots, X_{|N|})$
\EndIf
\EndIf
\EndProcedure
\end{algorithmic}
\end{algorithm}

\begin{lemma}\label{lem: com}
In line 14 of Algorithm~\ref{alg:02}, $X=(X_1,\ldots, X_{|N|})$ is a \PROPavg allocation to $N$.
\end{lemma}
\begin{proof}
Fix any $i\in N$. We show that $i$ is \PROPavg-satisfied by $X$.

\begin{description}
\item[Case 1:]  $i \in N_2$

In this case, agent $i$ receives exactly one good in the while statement.
By the while condition, we have 
\begin{align*}
v_i(X_i) &\ge \frac{1}{n}-\frac{1}{n-1}\sum_{j\in N_2} m_i(X_j)\\
& \ge \frac{1}{n}-\frac{1}{n-1}\sum_{j\in N\setminus i} m_i(X_j).
\end{align*}
Thus, $i$ is \PROPavg-satisfied by $X$.

\item[Case 2:] $i\in N_1$ and $l=1$

In this case, we have
\begin{align*}
v_i(X_i) =v_i(M_1)&= v_i(M)-\sum_{j\in N_2} m_i(X_j)\\
& =  \left(\frac{1}{n}-\frac{1}{n-1}\sum_{j\in N_2} m_i(X_j)\right)+\underbrace{\left(\frac{n-1}{n}-\frac{n-2}{n-1}\sum_{j\in N_2} m_i(X_j)\right)}_{\ge 0}\\
&\ge \frac{1}{n}-\frac{1}{n-1}\sum_{j\in N_2} m_i(X_j) \\
&= \frac{1}{n}-\frac{1}{n-1}\sum_{j\in N\setminus i} m_i(X_j), \\ 
\end{align*}
where the inequality follows from $\frac{n-1}{n}\ge \frac{n-2}{n-1}\ge \frac{n-2}{n-1}\sum_{j\in N_2} m_i(X_j)$.
\item[Case 3:]  $i \in N_1$ and $l\ge 2$

Since $(X_1,\ldots, X_l)$ is a \PROPavg allocation of $M_1$ to $N_1$, we have
\begin{align}\label{eq:1}
v_i(X_i) &\ge \frac{v_i(M_1)}{l} -\frac{1}{l-1}\sum_{j\in N_1\setminus i} m_i(X_j)\nonumber \\
&= \frac{1}{l}-\frac{1}{l}\sum_{j\in N_2} m_i(X_j)-\frac{1}{l-1}\sum_{j\in N_1\setminus i} m_i(X_j).
\end{align}
In line 13 of Algorithm~\ref{alg:02}, the while condition in line 7 does not hold for agent $i$.
Thus, it holds that 
\begin{equation}\label{eq:2}
m_i(X_j) < \frac{1}{n}-\frac{1}{n-1}\sum_{j'\in N_2} m_i(X_{j'})
\end{equation}
for any $j\in N_1\setminus i$.
Summing up  inequality~(\ref{eq:2}) for each $j\in N_1\setminus i$, we obtain
\begin{equation}\label{eq:3}
\sum_{j\in N_1\setminus i} m_i(X_j) < \frac{l-1}{n}-\frac{l-1}{n-1}\sum_{j\in N_2} m_i(X_j).
\end{equation}
By multiplying inequality~(\ref{eq:3}) by $\frac{n-l}{l(l-1)}>0$ and rearranging, we have
\begin{equation}\label{eq:4}
0  > -\frac{n-l}{ln}+\frac{n-l}{l(n-1)}\sum_{j\in N_2} m_i(X_j)+\frac{n-l}{l(l-1)}\sum_{j\in N_1\setminus i} m_i(X_j).
\end{equation}
Summing up inequalities~(\ref{eq:1}) and (\ref{eq:4}), we have 
\begin{equation}\label{eq:5}
v_i(X_i) > \frac{1}{n}+\left(-\frac{1}{l}+\frac{n-l}{l(n-1)}\right)\sum_{j\in N_2} m_i(X_j)+\left(-\frac{1}{l-1}+\frac{n-l}{l(l-1)}\right)\sum_{j\in N_1\setminus i} m_i(X_j).
\end{equation}
Since $2\le l < n$, 
by direct calculation, we have
\begin{equation*}
-\frac{1}{l}+\frac{n-l}{l(n-1)}+\frac{1}{n-1}= \frac{1}{l(n-1)} \ge 0
\end{equation*}
and
\begin{align*}
-\frac{1}{l-1}+\frac{n-l}{l(l-1)}+\frac{1}{n-1}&= \frac{1}{l(l-1)(n-1)}\left(-l(n-1)+(n-l)(n-1)+l(l-1)\right)\\
&=\frac{(n-l)(n-l-1)}{l(l-1)(n-1)} \\
& \ge 0.
\end{align*}

Applying these inequalities to inequality~(\ref{eq:5}), we finally obtain
\begin{align*}
v_i(X_i) &> \frac{1}{n}-\frac{1}{n-1}\left(\sum_{j\in N_2} m_i(X_j)+\sum_{j\in N_1\setminus i} m_i(X_j)\right).\\
&=\frac{1}{n}-\frac{1}{n-1}\sum_{j\in N\setminus i} m_i(X_j), 
\end{align*}
which implies that $i$ is \PROPavg-satisfied by $X$.
\end{description}
Therefore, $X$ is a \PROPavg allocation to $N$ in line 16 of Algorithm~\ref{alg:02}.
\end{proof}

We finally give the proof of Theorem~\ref{thm: main} by showing that Algorithm~\ref{alg:02} is a polynomial time algorithm to find a \PROPavg allocation.

\begin{proof}[Proof of Theorem~\ref{thm: main}]
We first show the correctness of Algorithm~\ref{alg:02}.
If $|N|=1$, our algorithm obviously returns a \PROPavg allocation in line 3.
Assume that $|N|\ge 2$.
If $|N_2|\ge 1$, it returns a \PROPavg allocation in line 14 by Lemma~\ref{lem: com}.
Otherwise, since the while condition in line 7 does not hold for any agent in $N_1$, $v_i(g)< \frac{1}{n}$ holds for any agent $i\in N_1$ and good $g\in M_1$ in line 16.
Thus, $X=(X_1, \ldots, X_n)$ satisfies (P1) by Lemma~\ref{lem: rec}, where $X_n$ is an empty set.
The rest of the algorithm finds an allocation to $V_2$ satisfying (P2) and returns a \PROPavg allocation as in Algorithm~\ref{alg:01}.
Therefore, Algorithm~\ref{alg:02} returns a \PROPavg allocation in all cases.

We finally show that Algorithm~\ref{alg:02} completes in time polynomial in the number of agents and items.
Let $T(n,m)$ be the worst case time complexity of Algorithm~\ref{alg:02} when $|N|=n$ and $|M|=m$.
Clearly, $T(1,m)=O(m)$.
We can check the while condition in line 7 and execute the body of the while loop in polynomial time of $n$ and $m$.
In addition, as mentioned at the beginning of Section~\ref{sec:poly}, 
Lines 17 to 22 can be executed in polynomial time of $n$ and $m$.
Thus, $T(n,m)$ can be expressed as
$$T(n,m)=O({\rm poly}(n,m))+\max\{\max_{\substack{1\le n' \le n-1\\ 1\le m' \le m-1}}T(n',m'), T(n-1, m)\}.$$
Therefore, $T(n,m)$ is polynomially bounded in $n$ and $m$.
\end{proof}

\section{Conclusion}\label{sec: discussion}
In this paper, we have introduced \PROPavg, which is a stronger notion than
\PROPm, and shown that a \PROPavg allocation always exists and can be computed in polynomial time when each agent has a non-negative additive valuation.
In order to devise our algorithm, we have developed a new technique that
generalizes the cut-and-choose protocol.
This technique is interesting by itself and seems to have a potential for
further applications.
In fact, we can define a bipartite graph like the \PROPavg-graph for
another fairness notion, and our argument works if we obtain an allocation
satisfying a (P2)-like condition.
We expect that this technique will be used in other contexts as well.

There are still several future work.
Whether \textsf{Avg-EFX}, which is a stronger notion than
\PROPavg exists for four or more agents is an interesting open problem.
Another direction is to consider weighted approximate proportionality.
In the weighted case, each agent $i$ has a non-negative weight $\alpha_i$, where $\alpha_1+\cdots + \alpha_n=1$.
The goal is to find an allocation $X$ such that $v_i(X_i)\ge \frac{\alpha_i}{n}-d_i(X)$ for each agent $i\in N$.
Whether weighted \textsf{PROPavg} allocation exists is also an interesting problem.

\subsection*{Acknowledgments}
This work was partially supported by the joint project of Kyoto University
and Toyota Motor Corporation, titled ``Advanced Mathematical Science for
Mobility Society'', and
by JSPS, KAKENHI grant number JP19H05485, Japan.
\bibliography{prop} 
\bibliographystyle{plain} 

\end{document}